\pdfoutput=1
\documentclass[
final
, nomarks
]{dmtcs-episciences}
\usepackage[utf8]{inputenc}

\usepackage{hyperref}
\usepackage{subfigure}
\usepackage{amsmath}
\usepackage{amsthm} % 
\usepackage{amssymb}
\usepackage[T1]{fontenc} % for \textsc inside \section
\usepackage{enumerate}
\usepackage{comment}
\usepackage{siunitx}

\usepackage{float}
\usepackage{dsfont}
\usepackage{tikz}
\usetikzlibrary{positioning}
\usetikzlibrary{decorations}
\usetikzlibrary{calc}
\usetikzlibrary{decorations.pathreplacing}
\usetikzlibrary{shapes.misc}
\usetikzlibrary{matrix}
\usetikzlibrary{math}
\usepackage{bookmark}

\usepackage[round]{natbib}
\usepackage{xurl}
\hypersetup{
    breaklinks=true,
  pdftitle = {Representing polynomial of st-connectivity},
  pdfauthor = {Janis Iraids, Juris Smotrovs},
  pdfstartview=FitH,
  unicode=true
}

\theoremstyle{plain}

\newcommand{\ket}[1]{\left| #1 \right\rangle}
\newcommand{\bra}[1]{\left\langle #1 \right|}

\newcommand{\norm}[1]{\left\| #1 \right\|}
\newcommand{\dc}{\textsc{AQ-Connectivity}}
\newcommand{\stc}{\textsc{ST-Connectivity}}
\newcommand{\faces}[1]{\mathcal{F}\left( #1 \right)}
\newcommand{\subsetdir}{\overset{\bullet}{\subset}}

\newtheorem{theorem}{Theorem}
\newtheorem{definition}{Definition}
\newtheorem{lemma}{Lemma}
\newtheorem{proposition}{Proposition}
\newtheorem{corollary}{Corollary}

\AtBeginDocument{% for small caps in math mode
  \DeclareFontShape{T1}{cmr}{m}{scit}{<->ssub*lmr/m/scsl}{}%
}

\author{J\={a}nis Iraids \and Juris Smotrovs}
\title{Representing polynomial of \stc}
\affiliation{Faculty of Computing, University of Latvia, Riga, Latvia}

\begin{document}
\publicationdata{vol. 25:2 }{2023}{22}{10.46298/dmtcs.9934}{2022-08-18; 2022-08-18; 2023-07-18}{2023-10-19}\maketitle

\begin{abstract}
We show that the coefficients of the representing polynomial of any monotone Boolean function are the values of the M\"obius function of an atomistic lattice related to this function. Using this we determine the representing polynomial of any Boolean function corresponding to a \stc\ problem in acyclic quivers (directed acyclic multigraphs). Only monomials corresponding to unions of paths have non-zero coefficients which are $(-1)^D$ where $D$ is an easily computable function of the quiver corresponding to the monomial (it is the number of plane regions in the case of planar graphs).
We determine that the number of monomials with non-zero coefficients for the two-dimensional $n \times n$ grid connectivity problem is $2^{\Omega(n^2)}$.
\end{abstract}

\section{Introduction}

In this paper we study the representing polynomials of Boolean functions. % corresponding to certain graph connectivity problems. 
Representations of Boolean functions by polynomials of various forms have a number of applications to computer
science, from circuit lower bounds \citep{Hastad86} 
to machine learning \citep{LinialMN93,BshoutyT96} 
and quantum query algorithms \citep{BealsBCMW01,BuhrmanW02}.
A detailed overview of mathematical properties of such polynomials and their applications can be found in the textbook by \citet{ODonnell2014}.

In the current paper we focus on exact representations of monotone Boolean functions, in particular, for Boolean functions that correspond to \stc.
The motivation of our work comes from quantum computing 
where subgraph connectivity problems emerge in the context of producing quantum speedups for various problems, 
such as Travelling Salesman Problem \citep{ABIKPV2019} and Edit Distance \citep{ABIKKPSSV2020}. 
It is well known that quantum algorithms can be described by polynomials \citep{BealsBCMW01} and this connection
has been used to prove a number of lower bounds on quantum algorithms. In particular, the degree of the representing polynomial is used to give lower bounds in the exact quantum query model (where the algorithm has to output the correct answer with certainty) and the minimum degree of an approximating polynomial is used for the much more natural bounded error quantum algorithms. In many cases, the polynomials lower bound is asymptotically tight and characterizes quantum query complexity up to a constant factor. Because of that, we think that it may be interesting to understand polynomials that represent the corresponding subgraph connectivity problems. Even though we focus on the representing polynomials, they can be used to obtain optimal approximating polynomials in some cases \citep{Beniamini2020}.

We now give a more technical overview of problems that we study and our results.
Every Boolean function $f:\{0,1\}^n\to\{0,1\}$ can be expressed as a real multilinear polynomial in a unique way. We will refer to it as the representing polynomial of the Boolean function. Such a polynomial is often used to estimate various complexity measures of a Boolean function, or to construct an algorithm related to it. 

The representing polynomials are often studied in the Fourier basis in which the bits $0$ and $1$ are replaced by values $1$ and $-1$, respectively. In particular, the cited textbook mostly deals with these polynomials in the Fourier basis. However, we will concentrate on these polynomials in the standard, $\{0,1\}$ basis, using approach similar to the one employed by Beniamini and Nisan in the recent papers \citep{BeniaminiNissan2020,Beniamini2020} involving lattices, their M\"obius functions and convex polytopes.

In this paper we deal with Boolean functions corresponding to the problem of \stc\ in acyclic quivers (\dc). In such a problem, the input bits correspond to the arcs of a given acyclic quiver, denoting presence ($1$) or absence ($0$) of an arc, and the Boolean function is equal to $1$ iff there is a path (consisting of present arcs) from a starting vertex $s$ (source) to a final vertex $t$ (sink). % References?? Can we cite our papers?
Instances of such problems are found in the Travelling Salesman Problem, where the graph is the Boolean hypercube (see, e.g., \citet{ABIKPV2019}), and in the Edit Distance problem where the graph is a two-dimensional grid and the task is to determine the shortest distance (see, e.g., \citet{ABIKKPSSV2020}).

Boolean functions corresponding to $\dc$ are monotone. In Section~\ref{sec:repmon} we show that the coefficients of the representing polynomial of any monotone Boolean function are the values of the M\"obius function (with minus sign) of the poset of unions of its prime implicants (Theorem~\ref{thm:poly}). This generalizes the corresponding
result by \citet[Proposition~3.4]{BeniaminiNissan2020}.
We also characterize the posets obtainable as unions of prime implicants: they are exactly the finite atomistic lattices (Proposition~\ref{prop:atomnec} and Theorem~\ref{thm:atomsuff}).

Then, relating an acyclic quiver $G$ to a flow polytope and using results from \citet{Hil03}, we compute the corresponding M\"obius function obtaining formula for the representing polynomial (Section~\ref{sec:moebdag}):
  \begin{equation*}p_G(x)=\sum_{H \in U(P_G)\setminus \{\varnothing\}}{(-1)^{D(H)}\prod_{i\in H}{x_i}}\end{equation*}
where the summation variable $H$ ranges over non-empty unions of paths from the source to the sink, and $D(H)$ is an easily computable function of $H$ (in the case of a planar graph it is essentially the number of regions in which $H$ divides the plane).

Since the number of monomials with non-zero coefficients is related to the communication complexity of the Boolean function (see, e.g., Section~4.4.1 in \citet{BeniaminiNissan2020}) and can be used to obtain a good approximating polynomial \citep[Lemma~31]{Beniamini2020}, we estimate it in the case when $G$ is a two-dimensional grid (Section~\ref{sec:sizegrid}). In Section~\ref{sec:posmoebpoly} we provide some preliminaries, and Section~\ref{sec:conc} contains the conclusion.

\section{Posets, their M\"obius functions, and convex polytopes}
\label{sec:posmoebpoly}

Let $O=\left\langle S, \leqslant_O \right\rangle$ be a poset over $S$ with order relation $\leqslant_O$. An antichain is any subset of $S$ consisting of mutually incomparable elements (under the order relation $\leqslant_O$). A join of $s\in S$ and $t\in S$ is an element $u\in S$ such that $\forall v:(s\leqslant_O v \text{ and } t\leqslant_O v)\rightarrow (u \leqslant_O v)$ if it exists. In a slight abuse of the notation we will denote the join of $s$ and $t$ by $s\cup t$. A poset for which every pair of elements have a join is called a join-semilattice. In a poset with least element $\varnothing$ an element $a$ is called an atom if $\varnothing <_O a$ and there is no $v$ such that $\varnothing <_O v <_O a$. A poset with least element $\varnothing$ is called atomistic if every element is a join of some set of atoms (join can be generalized to multiple elements by associativity). Consider the set inclusion poset $O=\left\langle 2^S,\subset \right\rangle$, and let $A$ be some antichain of this poset. Then let $U(A)$ be the induced poset of unions: $U(A)=\left\langle \{\bigcup_{b \in B}{b}|B\subseteq A\}, \subset \right\rangle$.

For every poset $O$ there exists a unique function called the M\"obius function $\mu:S\times S \rightarrow \reals$, with the following two properties:
\begin{enumerate}[1)]
    \item For all $s \in S$: $\mu(s,s)=1$, and
    \item For all $u,v \in S$ such that $u <_O v$: \[\sum_{s:u \leqslant_O s \leqslant_O v}{\mu(u,s)}=0.\]
\end{enumerate}

Let $H=\{x\in \reals^n| ax=b\}$ denote a hyperplane and let $H^+=\{x\in \reals^n| ax\geq b\}$ denote one of the half-spaces whose boundary is $H$. An intersection of half-spaces $B=H_1^+ \cap H_2^+ \cap \cdots \cap H_m^+$ that is bounded is called a convex polytope. If $H^+\cap B=B$, then the intersection $H\cap B$ is called a face of $B$; and each face itself is a polytope. Using the duality of linear programming one can show that
\begin{lemma}[{\citet[Chapter~8, Eqn.~(11)]{Schrijver1986}}]\label{lem:faces} A non-empty $F\subseteq \reals^n$ is a face of $H_1^+ \cap H_2^+ \cap \cdots \cap H_m^+$ if and only if $F=\bigcap_{i\in M}{H_i}\cap \bigcap_{i\notin M}{H_i^+}$ for some $M\subseteq [m]$.
\end{lemma}
In other words, a face corresponds to a system where some inequalities are replaced by equalities. This representation may be non-unique. Denote by $\faces{B}$ the set of faces of polytope $B$. Then $\left\langle \faces{B}, \subset \right\rangle$ is a poset that is a lattice called the face lattice of polytope $B$. For a polytope $B$ let its dimension be the largest $n$ such that there exist $v_1, v_2, \ldots, v_{n+1}\in B$ such that $v_1-v_{n+1},v_2-v_{n+1}, \ldots, v_n-v_{n+1}$ are linearly independent. Let us denote the dimension of a polytope $B$ by $\dim{B}$. Let $\dim{\varnothing}=-1$. Then
\begin{theorem}[Euler's relation, \cite{Gruenbaum2003,Broendsted1983}]
    For any two faces $F_1,F_2\in \faces{B}$, such that $F_1\subset F_2$:
  \begin{equation}\sum_{F:F_1\subseteq F\subseteq F_2}{(-1)^{\dim{F}}}=0.\end{equation}
\end{theorem}
Let $\mu_B$ denote the M\"obius function of the face lattice of the polytope $B$.
\begin{corollary}\label{cor:euler}
  \begin{equation}\mu_B(F_1, F_2)=(-1)^{\dim{F_2}-\dim{F_1}}.\end{equation}
\end{corollary}
\begin{proof}
  Since the M\"obius function is unique, it is sufficient to verify that $\mu_B$ as defined here satisfies the properties 1) and 2). Clearly, $\mu_B(F, F)=1$. For $F_1 \subset F_2$:
  \begin{equation}\sum_{F:F_1\subseteq F \subseteq F_2}{\mu_B(F_1, F)}=
  \sum_{F:F_1\subseteq F \subseteq F_2}{(-1)^{\dim{F}-\dim{F_1}}}=
  (-1)^{\dim{F_1}} \sum_{F:F_1\subseteq F \subseteq F_2}{(-1)^{\dim{F}}}=0.\end{equation}
\end{proof}

\section{Representing monotone functions}
\label{sec:repmon}

Let $f:\{0,1\}^n\rightarrow \{0,1\}$ be a Boolean function on $n$ variables. For two $n$-bit strings $x$ and $y$ we will denote by $x_i$ the $i$-th bit of the string, and say that $x\leq y$ if for all $i$: $x_i \leq y_i$. We say that a Boolean function is monotone if for all $x,y\in \{0,1\}^n$: $x\leq y$ implies $f(x)\leq f(y)$. We say that a set $I\subseteq [n]$ is a prime implicant of $f$ if the property $(\forall i\in I: x_i=1) \implies f(x)=1$ holds for $I$, but fails for every subset $I'\subset I$. Every monotone Boolean function has a unique representation in terms of the set of its prime implicants denoted by $P_f$. Since $P_f$ is an antichain in the inclusion poset of the subsets of $[n]$, we can consider its subposet $U(P_f)$.

\begin{proposition}\label{prop:atomnec}
The poset $\left\langle U(P_f), \subset \right\rangle$ is an atomistic lattice.
\end{proposition}

\begin{proof}
Since $\left\langle U(P_f), \subset \right\rangle$ is a subposet of the lattice
$\left\langle 2^{[n]}, \subset \right\rangle$ containing all the original joins 
(unions) of the elements of $U(P_f)$, it is an upper semilattice. Moreover, it is a
lattice, since $U(P_f)$ contains $\varnothing$ as the least element, and any finite
bounded upper semilattice is a lattice (see e.g. Proposition~3.3.1 in \citet{Stanley2012}).
Since the elements of $P_f$ are incomparable, and $U(P_f)$ contains beside them only
their unions (including the empty one), $P_f$ is the set of atoms of $U(P_f)$,
and every element of $U(P_f)$ is expressible as their join (union), 
i.e.\ the lattice $\left\langle U(P_f), \subset \right\rangle$ is atomistic.
\end{proof}

\begin{theorem}\label{thm:atomsuff}
Every finite atomistic lattice is isomorphic to
$\left\langle U(P_f), \subset \right\rangle$ for some monotone Boolean function $f$.
\end{theorem}

\begin{proof}
Let $\left\langle L, \leqslant_L \right\rangle$ be a finite atomistic lattice, and
let $A=\{a_1,a_2,\ldots,a_m\}\subset L$ be the set of its atoms.

We will prune a certain $(0,1)$-matrix $M$ to construct the prime implicants $P_f=\{b_1,b_2,\ldots,b_m\}$ 
of the corresponding monotone Boolean function $f$.

Initially let the matrix $M$ be of size $(2^{m}-1)\times2^{m}$, 
with columns indexed by subsets of $[m]$ and rows indexed by non-empty subsets of $[m]$.
For $S,T\subseteq[m]$, $S\neq\varnothing$ we define
\begin{equation}
M_{S,T}=\begin{cases}
0 & \mbox{if }S\cap T = \varnothing, \\
1 & \mbox{if }S\cap T \neq \varnothing.
\end{cases}
\end{equation}
We will delete some columns of this matrix to obtain its final form (and thus the correct set of prime implicants).

For each $i\in[m]$ the row $\{i\}$ will represent the prime implicant $b_i$ as a bit string.
More formally, to obtain $b_i$, let $n$ be the current number of columns,
then reindex the columns by numbers from $1$ to $n$, and let $b_i=\{j\in[n]\mid M_{\{i\},j}=1\}$.
For convenience, let $u:\;[n]\to 2^{[m]}$ map the new column indices to the old ones.
(Note that $b_i$, $n$, $u$ all change when we delete a column of $M$.)

The remaining rows represent unions of prime implicants as bit strings in a similar way.
In particular, row $S$ represents $\bigcup_{i\in S}b_i$. 
Indeed: $M_{S,j}=1\Leftrightarrow\exists i\,(i\in S\cap u(j))\Leftrightarrow
\exists i\in S\,(M_{\{i\},j}=1)\Leftrightarrow\exists i\in S\,(j\in b_i)\Leftrightarrow j\in\bigcup_{i\in S}b_i$. For the rest of the proof we revert to the original indexing of columns by subsets.

Initially $U(P_f)$ generated by $M$ is isomorphic to Boolean algebra of rank $m$.
(We don't need this fact explicitly, so we omit the proof.)
We need to make it isomorphic to $L$.
Since $L$ is atomistic, every one of its elements is expressible as a union of its atoms.
However, some unions of atoms may be equal between themselves.
Suppose an equality $\bigcup_{j\in S}a_j=\bigcup_{j\in T}a_j$ holds for some distinct
$S,T\subseteq[m]$. To achieve that a corresponding equality holds for $b_j$,
we remove from $M$ all columns which make these unions distinct, that is, all
columns $W$ such that $M_{S,W}\neq M_{T,W}$.

{\centering 
\noindent
\begin{figure}[H]
\begin{minipage}[b]{0.26\textwidth}
\centering
        \resizebox{\textwidth}{!}
        {
\begin{tikzpicture}[scale=0.85,yscale=0.85]
    \node (emptyset) at (0,0) {$\varnothing$};
    \node (a) at (-2,2) {$a_1$};
    \node (b) at (0,2) {$a_2$};
    \node (c) at (2,2) {$a_3$};
    \node (bc) at (1,4) {$a_2a_3$};
    \node (abc) at (0,6) {$a_1a_2a_3$};
    \draw (emptyset) -- (a);
    \draw (emptyset) -- (b);
    \draw (emptyset) -- (c);
    \draw (a) -- (abc);
    \draw (b) -- (bc);
    \draw (c) -- (bc);
    \draw (bc) -- (abc);
\end{tikzpicture}
}
\caption{Atomistic lattice with atoms $a_1$, $a_2$, $a_3$ and equalities $a_1a_2=a_1a_3=a_1a_2a_3$\label{fig:atomistic}}
\end{minipage}
\hspace{0.03\textwidth}
    \begin{minipage}[b]{0.69\textwidth}
        \centering
        \resizebox{\textwidth}{!}
        {

        \begin{tikzpicture}[every node/.style={text depth=.5ex,text height=2ex}]
        
        \matrix (m) [matrix of math nodes, ampersand replacement=\&,
        column 1/.style={text width=8ex, align=center}, column 2/.style={text width=8ex, align=center}, column 4/.style={text width=3.5ex, align=center}, column 5/.style={text width=3.5ex, align=center}, column 6/.style={text width=3.5ex, align=center}, column 10/.style={text width=8ex, align=center}
        ]
        {
             \, \& \, \& \varnothing \& \{1\} \& |[fill=gray!30]|\{2\} \& |[fill=gray!30]|\{3\} \& \{1,2\} \& \{1,3\} \& \{2,3\} \& \{1,2,3\} \\
             \{1\} \& b_1 \& 0 \& 1 \& 0 \& 0 \& 1 \& 1 \& 0 \& 1\\
             \{2\} \& b_2 \& 0 \& 0 \& 1 \& 0 \& 1 \& 0 \& 1 \& 1\\
             \{3\} \& b_3 \& 0 \& 0 \& 0 \& 1 \& 0 \& 1 \& 1 \& 1\\
             |[fill=gray!30]|\{1,2\} \& |[fill=gray!30]|b_1b_2 \& 0 \& 1 \& |[fill=gray!30]|1 \& |[fill=gray!30]|0 \& 1 \& 1 \& 1 \& 1\\
             |[fill=gray!30]|\{1,3\} \& |[fill=gray!30]|b_1b_3 \& 0 \& 1 \& |[fill=gray!30]|0 \& |[fill=gray!30]|1 \& 1 \& 1 \& 1 \& 1\\
             \{2,3\} \& b_2b_3 \& 0 \& 0 \& 1 \& 1 \& 1 \& 1 \& 1 \& 1\\
             |[fill=gray!30]|\{1,2,3\} \& |[fill=gray!30]|b_1b_2b_3 \& 0 \& 1 \& |[fill=gray!30]|1 \& |[fill=gray!30]|1 \& 1 \& 1 \& 1 \& 1\\
        };
        \draw (m-1-1.south west) -- (m-1-10.south east);
        \draw (m-4-1.south west) -- (m-4-10.south east);
        \draw (m-1-2.north east) -- (m-8-2.south east);
        \end{tikzpicture}

        }
    \caption{Initial state of the matrix $M$. In addition to row indices the corresponding unions of prime implicants are specified. The rows $b_1b_2$, $b_1b_3$ and $b_1b_2b_3$ must be equal, therefore columns $\{2\}$ and $\{3\}$ making them different must be deleted.\label{fig:matrixMstart}}
    \end{minipage}

    \begin{minipage}[t]{0.51\textwidth}
        \vspace{0pt}
        \centering
        \resizebox{\textwidth}{!}{

        \begin{tikzpicture}[every node/.style={text depth=.5ex,text height=2ex}]
        
        \matrix (m) [matrix of math nodes, ampersand replacement=\&,
        column 1/.style={text width=6ex, align=center}, column 7/.style={text width=8ex, align=center}
        ]
        {
             \, \& \varnothing \& \{1\} \& \{1,2\} \& \{1,3\} \& \{2,3\} \& \{1,2,3\} \\
             b_1 \& 0 \& 1 \& 1 \& 1 \& 0 \& 1\\
             b_2 \& 0 \& 0 \& 1 \& 0 \& 1 \& 1\\
             b_3 \& 0 \& 0 \& 0 \& 1 \& 1 \& 1\\
             b_1b_2 \& 0 \& 1 \& 1 \& 1 \& 1 \& 1\\
             b_1b_3 \& 0 \& 1 \& 1 \& 1 \& 1 \& 1\\
             b_2b_3 \& 0 \& 0 \& 1 \& 1 \& 1 \& 1\\
             b_1b_2b_3 \& 0 \& 1 \& 1 \& 1 \& 1 \& 1\\
        };
        \draw (m-1-1.south west) -- (m-1-7.south east);
        \draw (m-4-1.south west) -- (m-4-7.south east);
        \draw (m-1-1.north east) -- (m-8-1.south east);
        \end{tikzpicture}

        }
    \caption{State of the matrix $M$ after the deletion of columns $\{2\}$ and $\{3\}$. Now the lattice generated by prime implicants $b_1,b_2,b_3$ is isomorphic to the lattice of Figure \ref{fig:atomistic}.\label{fig:matrixMend}}
    \end{minipage}
    \hspace{0.03\textwidth}
    \begin{minipage}[t]{0.45\textwidth}
        \vspace{0pt}
        \centering
        \resizebox{0.85\textwidth}{!}{
        \begin{tikzpicture}[every node/.style={text depth=.5ex,text height=2ex}]
        
        \matrix (m) [matrix of math nodes, ampersand replacement=\&,
        column 1/.style={text width=6ex, align=center}, column 5/.style={text width=6ex, align=center}
        ]
        {
             \, \&  \{1\} \& \{1,2\} \& \{1,3\} \& \{2,3\} \\
             b_1 \& 1 \& 1 \& 1 \& 0\\
             b_2 \& 0 \& 1 \& 0 \& 1\\
             b_3 \& 0 \& 0 \& 1 \& 1\\
             b_1b_2 \& 1 \& 1 \& 1 \& 1\\
             b_1b_3 \& 1 \& 1 \& 1 \& 1\\
             b_2b_3 \& 0 \& 1 \& 1 \& 1\\
             b_1b_2b_3 \& 1 \& 1 \& 1 \& 1\\
        };
        \draw (m-1-1.south west) -- (m-1-5.south east);
        \draw (m-4-1.south west) -- (m-4-5.south east);
        \draw (m-1-1.north east) -- (m-8-1.south east);
        \end{tikzpicture}
        }
    \caption{The matrix $M$ after removing columns which are unions of other columns: $\{1, 2, 3\}=\{1\}\cup\{2,3\}$ and $\varnothing$ (the empty union). Lattice generated by $\{b_1,b_2,b_3\}$ does not change. A monotone Boolean function with lattice isomorphic to that of Figure \ref{fig:atomistic} must have at least $4$ input bits, at least one per each of the remaining columns.\label{fig:matrixMopt}}
    \end{minipage}
\end{figure}
}

We perform this operation for every equality that holds between the unions of atoms $a_j$
ensuring that the corresponding equality holds also for the unions of $b_j$.
This process is illustrated in Figures~\ref{fig:atomistic}--\ref{fig:matrixMend} for an atomistic lattice with three 
atoms $a_1$, $a_2$, $a_3$ and identities $a_1\cup a_2 = a_1\cup a_3= a_1\cup a_2\cup a_3$ (we omit the operation signs in the Figures for the sake of brevity).

%\bigskip
We claim that the matrix $M$ remaining after these operations generates the lattice $U(P_f)$ which we need.
To show that, it remains to prove that we haven't introduced any undesired equality
among unions of $b_j$ by removing too many columns.

Suppose the contrary: that for some $S,T\subseteq[m]$ we have 
$\bigcup_{j\in S}a_j\neq\bigcup_{j\in T}a_j$, but $\bigcup_{j\in S}b_j=\bigcup_{j\in T}b_j$.
Here among different $S$ that give the same union $\bigcup_{j\in S}a_j$, let us use
the maximum $S$ (i.e.\ such that no superset of $S$ gives the same union), similarly
let us use the maximum $T$. Since an equality among unions of $a_j$ implies an equality
among the corresponding unions of $b_j$, $\bigcup_{j\in S}b_j=\bigcup_{j\in T}b_j$ holds
also for the new $S$ and $T$, if we replaced any of them.

At least one of the sets $S\setminus T$ and $T\setminus S$ is nonempty; WLOG suppose
that $S\setminus T\neq\varnothing$.
The equality $\bigcup_{j\in S}b_j=\bigcup_{j\in T}b_j$ means that, among others, we have
removed the column $[m]\setminus T$ from $M$ because otherwise it would make these unions distinct:
$S\cap([m]\setminus T)=S\setminus T$, so $M_{S,[m]\setminus T}=1$ while $M_{T,[m]\setminus T}=0$.
Let the equality which caused the removal of column $[m]\setminus T$ be
$\bigcup_{j\in S'}a_j=\bigcup_{j\in T'}a_j$ where 
$M_{S',[m]\setminus T}=1$ and $M_{T',[m]\setminus T}=0$.
Then $T'\subseteq T\subset S'\cup T$ and 
$\bigcup_{j\in T}a_j=\bigcup_{j\in (T\setminus T')\cup T'}a_j=\bigcup_{j\in (T\setminus T')\cup S'}a_j$.
Note that in an atomistic lattice, if $\bigcup_{j\in X}a_j=\bigcup_{j\in Y}a_j$, then 
$\bigcup_{j\in X}a_j=\bigcup_{j\in Y}a_j=\bigcup_{j\in X\cup Y}a_j$, thus the last equality
implies $\bigcup_{j\in T}a_j=\bigcup_{j\in T\cup S'}a_j$ which contradicts the maximality of $T$.

Thus no undesired equality was introduced, and $L$ is isomorphic to $U(P_f)$ by mapping
that for each $S\subseteq[m]$ maps $\bigcup_{j\in S}a_j$ to $\bigcup_{j\in S}b_j$.
\end{proof}

\textbf{[Note.} The matrix obtained at the end of the process described in this proof is not optimal in terms of size:
any column with index expressible as a union of indices of some other remaining columns (including the empty union)
can be removed, since that does not introduce any new equality (see Figure \ref{fig:matrixMopt}). 
Indeed, if columns $W_1,\ldots,W_k$ and
$W=W_1\cup\cdots\cup W_k$ remain, then any inequality of rows witnessed in the column $W$: $M_{S,W}=0\neq 1=M_{T,W}$
is witnessed also in one of the columns $W_i$: $M_{S,W}=0$ implies $M_{S,W_i}=0$ for all $i$ and $M_{T,W}=1$
implies $M_{T,W_i}=1$ for at least one $i$. After these redundant columns are removed, all other columns must remain.
Indeed, if a column $W$ is still present, then it is not a union of remaining columns, thus there exists an element
$i\in W$ not belonging to any index $V$ of a remaining column such that $V\subset W$. 
Then the inequality of rows $[m]\setminus W$ and $([m]\setminus W)\cup\{i\}$ 
is witnessed only in this column: $M_{[m]\setminus W,W}=0\neq 1=M_{([m]\setminus W)\cup\{i\},W}$
while for any $V\subset W$: $M_{[m]\setminus W,V}=M_{([m]\setminus W)\cup\{i\},V}=0$, and for any other $V$:
$M_{[m]\setminus W,V}=M_{([m]\setminus W)\cup\{i\},V}=1$. Removing this column would make these rows equal, but
any rows that remained different after the process described in the proof of the theorem must remain different,
since they correspond to different elements of the lattice. Since the columns of the matrix $M$ correspond to input
bits of the Boolean function $f$, the number of the remaining columns determines the minimum number of input bits
of a monotone Boolean function corresponding to the given atomistic lattice.\textbf{]}

Let $p:\reals^n\rightarrow \reals$ be a multilinear polynomial. We can write $p$ as a linear combination of monomials:
\begin{equation}p(x_1, x_2, \ldots, x_n)=\sum_{S\subseteq [n]}{\alpha_S\prod_{i\in S}{x_i}}.\end{equation}
We say that a real multilinear polynomial $p$ represents a Boolean function $f$ if $p(x)=f(x)$ for all $x\in\{0,1\}^n$. Every Boolean function has a unique polynomial $p$ that represents it. The poset $U(P_f)$ and its associated M\"obius function $\mu_{U(P_f)}$ have a crucial role in determining the coefficients $\alpha_S$:

\begin{theorem}\label{thm:poly}
  Let $f$ be a monotone Boolean function not identical to $1$. Then its representing polynomial is
  \begin{equation}p_f(x_1, x_2, \ldots, x_n)=\sum_{S\in U(P_f)\setminus \{\varnothing\}}{-\mu_{U(P_f)}(\varnothing, S)}{\prod_{i\in S}{x_i}}.\end{equation}
\end{theorem}

\begin{proof}
We will show that the polynomial $p_f$ equals $f$ on all $x\in \{0,1\}^n$. Clearly, when $f(x)=0$ then for each $S\in U(P_f)\setminus\{\varnothing\}$ there exists an $i$ such that $x_i=0$, and so $p_f(x)=0$. If $f(x)=1$, then let $S$ be such that $x_i=1 \iff i\in S$. Let $S'=\cup \{I | I\in P_f, I\subseteq S\}$. $S'$ cannot be empty since $f(x)=1$. Then
\begin{equation*}p_f(x)=\sum_{T\in U(P_f):T \subseteq S'}{\alpha_T}=\sum_{T\in U(P_f):T \subseteq S'}{-\mu_{U(P_f)}(\varnothing,T)}+\mu_{U(P_f)}(\varnothing,\varnothing)=1.\end{equation*}
\end{proof}

\section{M\"obius function of \dc}
\label{sec:moebdag}
  Let $G=\left\langle V, E, s, t\right\rangle$ be a directed multigraph without cycles, namely, an acyclic quiver, where $V$ is the set of vertices, $E$ is the set of arcs with $s:E\rightarrow V$ and $t:E\rightarrow V$ defining the source and target of an arc. For our purposes we can fix $V$, $s$, $t$ and associate $G$ with the set of its arcs $E$. Thus we will write $H\subseteq G$ when $H$ is a subquiver of $G$, etc. Let $S(G) \subseteq V$ be the set of vertices with no incoming arcs -- sources, and let $T(G)\subseteq V$ be vertices with no outgoing arcs -- sinks.
\begin{definition}
 In $\dc_G$ we are given a subquiver $H$ of some fixed quiver $G$ as a set of bits defining its arcs $x_e=\begin{cases}1\text{ if }e\in H\\0\text{ if }e\notin H \end{cases}$, and our task is to determine if there is a non-empty path from a vertex in $S(G)$ to a vertex in $T(G)$ using only the arcs in $H$.
\end{definition}

Since we allow multigraphs, i.e., graphs which can have multiple arcs $e_1, \ldots, e_m$ with the same sources and targets $s(e_1)=\cdots =s(e_m)$, $t(e_1)=\cdots = t(e_m)$, for our purposes we can assume that there is exactly one source and one sink, because we can merge all sources into one (and similarly --- all sinks) without affecting connectivity. Henceforth, we will denote the unique source as $s$ and the unique sink as $t$. 

Given an acyclic quiver $G$, let $P_G$ be the finite set of paths connecting sources and sinks.  Note that $P_G$ has to be an antichain in the poset of subquivers of $G$. The following two theorems are a special case of Theorem~3.2 from \citet{Hil03}. We give our proofs for completeness.

\begin{theorem}\label{thm:iso}
  For all acyclic quivers $G$ the poset $U(P_G)$ is isomorphic to a face lattice of a convex polytope.
\end{theorem}

\begin{proof}
  For an acyclic quiver $G$ and its subquiver $H\subseteq G$, let $F(H)$ be the set of unit flows from $s$ to $t$, i.e., $F(H)$ consists of all vectors in $\reals^{|G|}$
  \[F(H) = \{(f_e)_{e\in G}\}_{f\in \reals^{|G|}},\]
  such that the following additional constraints are satisfied:
  \begin{enumerate}[(a)]
    \item Flow is non-negative: $\forall e\in G: f_e\geq 0$;
    \item \label{li:total}The total flow is 1: $\sum_{\substack{e:e\in G\\s(e)=s}}{f_e}=1$;
    \item \label{li:con}Flow conservation: $\forall v\notin \{s,t\}: \sum_{\substack{e:e\in G\\t(e)=v}}{f_e}=\sum_{\substack{e:e\in G\\s(e)=v}}{f_e}$.
    \item \label{li:sub}Flow is restricted to the subquiver: $\forall e\notin H: f_e=0$;
  \end{enumerate}
  Next, we show that $F:U(P_G)\rightarrow 2^{\reals^{|G|}}$ is indeed the bijection we sought. First, $F(H)$ shares almost all constraints with $F(G)$ except equalities (\ref{li:sub}). By Lemma~\ref{lem:faces} the faces of $F(G)$ correspond to the system of $F(G)$ where some inequalities are replaced by equalities, i.e., $f_e=0$ for some $H'\subseteq G$. Thus, $F(H)$ is a face of $F(G)$ for any $H\in U(P_G)$ since the inequalities $f_e\geq 0$ are replaced by equalities $f_e=0$ for $e \notin H$. On the other hand, for a subquiver $H'$ there exists a $H\in U(P_G)$ such that $H=\bigcup{\{p\in P_G|p \subseteq H'\}}$. But subquivers $H$ and $H'$ correspond to the same face; since all $e\in H'\setminus H$ have no path in $H'$ containing them, the flow on these arcs must be zero: $f_e=0$ for $e\in H'\setminus H$.  By construction the inclusion property is obviously preserved, since the flow is more constrained on a subquiver.
\end{proof}

Next we give a simple formula for computing the dimension of a face of the flow polytope. Let $D(H)=|H|-|\{s(e)|e\in H\} \cup \{t(e)|e\in H\} \setminus \{s,t\}|-1$.

\begin{theorem}\label{thm:dim}
  If $H\in U(P_G)$, then the dimension of the corresponding face is
  \begin{equation}\dim{F(H)}=D(H).\end{equation}
\end{theorem}
\begin{proof}
  First, the lemma is clearly true for empty quiver and a single path from $P_G$.

  Denote by $A\subsetdir B$, if $A,B \in U(P_G)$, $A\subset B$ and $\neg \exists C\in U(P_G): A\subset C \subset B$. Let us show that for all $H, H'\in U(P_G)$ such that $H \subsetdir H'$: $\Delta:=H'\setminus H$ is an ear of $H'$, i.e., a path $(v_0, v_1, \ldots, v_k)$ whose internal vertices $v_1, v_2,\ldots, v_{k-1}$ have no other adjacent arcs in $H'$. Hence $D(H')=D(H)+1$. Let us prove by contradiction assuming that $\Delta$ is not an ear. Clearly, there exists some subset $\delta \subseteq \Delta$ that is an ear of $\delta \cup H$; one can start a path with any arc of $\Delta$ and extend the path by traveling backwards and forwards along arcs in $\Delta$ until a vertex with an adjacent arc in $H$ is encountered. Consider the initial vertex $v_{init}$ of $\delta$. There must be a path $\alpha \subseteq H$ [potentially empty] from $s$ to $v_{init}$. If $v_{init}$ has no incoming arcs in $H$ then $v_{init}=s$. If $v_{init}$ has an incoming edge in $H$ then $\alpha$ is a path from $s$ to $v_{init}$. Symmetrically reasoning we obtain a path $\omega$ from the final vertex to $t$. Clearly, $\alpha \cup \delta \cup \omega \in P_G$ and so $\delta \cup H \in U(P_G)$. Thus $H \subset H\cup \delta \subset H\cup\Delta =H'$ --- a contradiction.

\begin{figure}[H]
    \centering
%    \resizebox{\textwidth}{!}{
        \begin{tikzpicture}[ns/.style={inner sep=0,outer sep=0}]
        \node[ns] (N1) at (0,0) {};
        \node[ns] (N2) at (0,1) {};
        \node[ns] (N3) at (1,0) {};
        \node[ns] (N4) at (1,2) {};
        \node[ns] (N5) at (2,1) {};
        \node[ns] (N6) at (2,2) {};
        \draw[-stealth,red,thick] (N1) -- (N2);
        \draw[-stealth,red,thick] (N2) -- (N4);
        \draw[-stealth,red,thick] (N4) -- (N6);
        \draw[-stealth,green,thick] (N1) -- (N3);
        \draw[-stealth,green,thick] (N3) -- (N4);
        \draw[-stealth,blue,thick] (N3) -- (N5);
        \draw[-stealth,blue,thick] (N5) -- (N6);
        \draw[-stealth,rounded corners=4pt]  ($(N1.center)+(-0.1,0)$) -- ($(N2.center)+(-0.1,0.1)$) -- node[midway,above left] {$p_0$} ($(N4.center)+(-0.1,0.1)$) -- ($(N6.center)+(0,0.1)$);
        \draw[-stealth,rounded corners=4pt]  ($(N1.center)+(0.1,0.1)$) -- ($(N3.center)+(-0.1,0.1)$) -- node[midway,left] {$p_1$}($(N4.center)+(-0.1,-0.1)$) -- ($(N6.center)+(-0.1,-0.1)$);
        \draw[-stealth,rounded corners=4pt]  ($(N1.center)+(0,-0.1)$) -- ($(N3.center)+(0.1,-0.1)$) -- node[midway,below right] {$p_2$}($(N5.center)+(0.1,-0.1)$) -- ($(N6.center)+(0.1,0)$);
        \matrix[anchor=north west,fill=white,draw,inner sep=0.3em,nodes={anchor=center, align=center}] at (4,2){
    \draw[-,red,thick] (0,0) -- (1,0);&\node[text width=1cm]{Ear 0}; \\
    \draw[-,green,thick] (0,0) -- (1,0);&\node[text width=1cm]{Ear 1}; \\
    \draw[-,blue,thick] (0,0) -- (1,0);&\node[text width=1cm]{Ear 2}; \\};
    \end{tikzpicture}
%    }
    \caption{Ear decomposition}\label{f:ears}
\end{figure}

  Each $H\in U(P_G)$ has an ``ear decomposition'' described as $\varnothing =H_{-1} \subsetdir H_0 \subsetdir H_1 \subsetdir \cdots \subsetdir H_{D(H)-1} \subsetdir H_{D(H)}=H$. In particular, let the path $\alpha \cup \delta \cup \omega \in P_G$ added from $H_{k-1}$ to $H_k$ be denoted by $p_k$. Then $H_k=H_{k-1}\cup p_k$.

  Let $\mathds{1}(p)=(f_e)_{e\in G}$ denote the unit flow along path $p$: $f_e=\begin{cases}1\text{ if }e\in p \\ 0 \text{ if } e \notin p\end{cases}$. Any flow satisfying all but the non-negativity constraints (i.e., satisfying (\ref{li:total}),(\ref{li:con}),(\ref{li:sub})) for $H$ is an affine combination of $\{\mathds{1}(p_i)|i\in \{0, 1, \ldots, D(H)\}\}$. We can establish this by induction on $D(H)$. This is obviously the case for $D(H)=-1$ and $D(H)=0$. Assuming that it is true for $D(H)=k-1$, $k>0$, let $f$ be the flow vector. Note that the flow $f_e$ for arcs $e$ in the ear $e \in H_{D(H)}\setminus H_{D(H)-1}$ must be equal. Consider the flow
  \begin{equation}f'=f-f_e(\mathds{1}(p_k)-\mathds{1}(p_{k-1})).\end{equation}
  $f'$ satisfies constraints (\ref{li:total}),(\ref{li:con}),(\ref{li:sub}) for $H_{D(H)-1}$ and so by inductive assumption $f'$ is in the affine span of $\{\mathds{1}(p_i)|i\in \{0, 1, \ldots, D(H)-1\}\}$. We conclude that $\dim{F(H)}\leq |\{p_0, p_1, \ldots, p_{D(H)}\}|-1 = D(H)$.

  $\{\mathds{1}(p_i)-\mathds{1}(p_0)|i\in \{1, \ldots, D(H)\}\}$ are linearly independent because each $\mathds{1}(p_i)-\mathds{1}(p_0)$ is outside the linear span of $\{\mathds{1}(p_i)-\mathds{1}(p_0)|i\in \{1, \ldots, i-1\}\}$. Obviously, $\{\mathds{1}(p_i)|i\in \{0, 1, \ldots, D(H)\}\}$ belong to $F(H)$. Therefore $\dim{F(H)}\geq D(H)$.
  \end{proof}
%
%  Another proof of Theorem~\ref{thm:dim} can be found in \cite{BV08}.

\begin{corollary}
  The unique multilinear polynomial representing $\dc_G$ is:
  \begin{equation}p_G(x)=\sum_{H \in U(P_G)\setminus \{\varnothing\}}{(-1)^{D(H)}\prod_{i\in H}{x_i}}.\end{equation}
\end{corollary}
\begin{proof}
    By Theorem~\ref{thm:poly} we have:
    \begin{equation}p_G(x)=\sum_{H\in U(P_G)\setminus \{\varnothing\}}{-\mu_{U(P_G)}(\varnothing, H)}{\prod_{i\in H}{x_i}}.\end{equation}
    By Theorem~\ref{thm:iso} the poset $U(P_G)$ is isomorphic to the face lattice of a polytope, and by the Corollary~\ref{cor:euler} of Euler's relation:
    \begin{equation}p_G(x)=\sum_{H\in U(P_G)\setminus \{\varnothing\}}{-(-1)^{\dim{F(H)-\dim{\varnothing}}}}{\prod_{i\in H}{x_i}}=
    \sum_{H\in U(P_G)\setminus \{\varnothing\}}{(-1)^{\dim{F(H)}}}{\prod_{i\in H}{x_i}}.\end{equation}
    Finally, by Theorem~\ref{thm:dim} we conclude that
    \begin{equation}p_G(x)=\sum_{H\in U(P_G)\setminus \{\varnothing\}}{(-1)^{D(H)}}{\prod_{i\in H}{x_i}}.\end{equation}
\end{proof}

\begin{corollary}
  The degree of $p_G(x)$ is maximal: $\deg p_G(x)=|G|$.
\end{corollary}
\begin{proof}
Since the whole quiver $G$ is also a union of paths: $G\in U(P_G)$, the coefficient at its monomial is $(-1)^{D(G)}$, i.e.~not zero.
\end{proof}

\section{Size of \texorpdfstring{$U(P_G)$}{U(P\_G)} for 2D grids}
\label{sec:sizegrid}

Let $G_n$ be a directed grid: the vertices of this quiver are labeled by $\{0,1,\ldots, n\}^2$ and there is an arc from vertex $(i_1, j_1)$ to $(i_2, j_2)$ iff $(i_1=i_2 \wedge j_2=j_1+1)$ or $(j_1=j_2 \wedge i_2=i_1+1)$.
\begin{theorem}\label{thm:gridlb}
$\left|U(P_{G_n})\right| \in 2^{\Omega(n^2)}.$
\end{theorem}
\begin{proof}
Consider the subgraph $H$ of grid consisting of all the horizontal edges and vertical edges with $i_1\in \{0, n\}$ (see Figure \ref{f:horizontal}).
\begin{figure}[H]
    \centering
        \resizebox{0.5\textwidth}{!}{
            \begin{tikzpicture}
                \tikzmath{
                    \n = 10;
                    \nmo = \n-1;
                }
                \foreach \x in {0,1,...,\n}
                    \foreach \y in {0,1,...,\nmo}
                    {
                        \draw[-stealth] (\y,\x) -- (\y+1,\x);
                    }
                \foreach \y in {0,1,...,\nmo}
                {
                    \draw[-stealth] (0,\y) -- (0,\y+1);
                    \draw[-stealth] (\n,\y) -- (\n,\y+1);
                }
                \foreach \x in {1,...,\nmo}
                    \foreach \y in {0,1,...,\nmo}
                    {
                        \draw[-stealth, dotted] (\x,\y) -- (\x,\y+1);
                    }
            \end{tikzpicture}
        }
        \caption{The subgraph $H$ of grid $G_n$ for $n=10$}\label{f:horizontal}
 \end{figure} 
Clearly, all the subgraphs containing $H$ are unions of paths since $H\in U(P_{G_n})$ and for any edge $e \notin H$ there is a path consisting of $e$ and only edges in $H$. The number of such subgraphs is $2^{(n+1)(n-1)}$.
\end{proof}

Unfortunately, this shows that the approximating polynomial produced by the construction scheme described in \citet[Lemma~31]{Beniamini2020} either has degree $\Omega(n^2)$ or the estimate of its degree is not optimal.

Let $D^{\operatorname{AND}}(f)$ denote the decision tree complexity of computing a function $f$ using a decision tree whose nodes are allowed to compute an $\operatorname{AND}$ of an arbitrary subset of input bits. Then, in light of \cite[Lemma~3.15]{BeniaminiNissan2020} stating that $D^{\operatorname{AND}}(f)\geq \log_3{|mon(f)|}$ where $mon(f)$ is the set of monomials with non-zero coefficients, we conclude that
\begin{corollary} $D^{\operatorname{AND}}(\dc_{G_n}) = \Omega(n^2)$.
\end{corollary}

\section{Conclusion}
\label{sec:conc}
Like \citet{BeniaminiNissan2020}, we studied also the dual $f^*$ of the function: $f^*(x_1,x_2,\ldots, x_n)=\neg f(\neg x_1, \neg x_2, \ldots, \neg x_n)$. Let $\dc^*(x_1, x_2,\ldots, x_n)=\sum_{S\subseteq [n]}{\alpha^*_S\prod_{i\in S}{x_i}}$. Prime implicants of $\dc^*$ are the minimal cuts. Even though we did not manage to prove it for arbitrary quivers, we found that for the quivers we analyzed the following properties are true:
\begin{itemize}
    \item $\alpha^*_S\in \{-1,0,1\}$;
    \item The sets $S$ corresponding to the non-zero $\alpha^*_S$ together with the empty set constitute an Eulerian lattice with the usual subset inclusion relation. However, unlike for $\dc$ the elements of this lattice are not all unions of minimal cuts, but a subset of them. In particular, a slight generalization of Lemma~4.6 from \citet{BeniaminiNissan2020} to any function with the union of prime implicants $U(P_f)$ being an Eulerian lattice holds. However, for $\dc$ it is not a sufficient criterion to determine whether $\alpha^*_S=0$.
\end{itemize}

The following topics could be of interest for future research:
\begin{itemize}
\item
What other useful classes of monotone Boolean functions, besides the Bipartite Perfect Matching \citep{BeniaminiNissan2020} and \dc, have simple M\"obius functions?
\item
The lattices of Bipartite Perfect Matching and $\dc$ are Eulerian implying simplicity coefficients of the representing polynomial. Is there a simple (good, useful) characterization of the monotone Boolean functions whose lattices are Eulerian?
\item
Can the representing polynomial of some $\dc$ problem be used to improve its complexity estimations? For instance, the current quantum query complexity estimations for the $n\times n$ two-dimensional grid still have gap between $\Omega(n^{1.5})$ and $O(n^2)$ \citep{ABIKKPSSV2020}. Since the quantum query complexity is lower bounded by the minimum degree of an approximating polynomial (divided by 2), one of the questions is: can the representing polynomial be used to obtain lower bounds exceeding $\Omega(n^{1.5})$ for the degree of an approximating polynomial?
\end{itemize}

\section*{Acknowledgements}
The authors wish to thank Andris Ambainis for useful suggestions on how to improve the paper. This research was supported by QuantERA ERA-NET Cofund in Quantum Technologies implemented within the European Union's Horizon 2020 Programme (QuantAlgo project), ERDF project 1.1.1.5/18/A/020 ``Quantum algorithms: from complexity theory to experiment'' and the Latvian Quantum Initiative under European Union Recovery and Resilience Facility project no. 2.3.1.1.i.0/1/22/I/CFLA/001.
\nocite{*}
\bibliographystyle{abbrvnat}
% use the following instead if you encounter problems 
%\bibliographystyle{alpha}
\bibliography{global}
\label{sec:biblio}

\end{document}